\documentclass[11pt,a4paper]{article}

\usepackage{authblk}
\usepackage[T1]{fontenc}
\usepackage[utf8]{inputenc}
\usepackage{lmodern}
\usepackage{microtype}
\usepackage[a4paper,margin=1.3in]{geometry}

\usepackage{amsmath,amssymb,amsthm,mathtools,bm}
\usepackage[hidelinks]{hyperref}

\newtheorem{prop}{Proposition}

\newenvironment{keywords}{%
  \vspace{0.5em}\noindent\textbf{Keywords:}\ }%
  {\par\vspace{0.75em}}

\newcommand{\ket}[1]{\lvert #1\rangle}
\newcommand{\bra}[1]{\langle #1\rvert}
\newcommand{\braket}[2]{\langle #1 \mid #2\rangle}
\newcommand{\qop}{\mathbin{\oplus}}
\newcommand{\C}{\mathbb{C}}

\title{\textbf{Comment on arXiv:2601.04248v1:\\
``Superposition of states in quantum theory'' (J.-M.\ Vigoureux)}}

\author[1]{Miko{\l}aj Sienicki}
\author[2]{Krzysztof Sienicki\thanks{E-mail: \texttt{niskrissienicki@gmail.com}}}

\affil[1]{Polish--Japanese Academy of Information Technology, ul.~Koszykowa~86, 02--008 Warsaw, Poland, European Union}
\affil[2]{Chair of Theoretical Physics of Naturally Intelligent Systems (NIS), Lipowa~2/Topolowa~19, 05--807 Podkowa Le\'sna, Poland, European Union}

\date{\today}

\begin{document}
\maketitle

\textbf{Abstract.}
Vigoureux suggests replacing the usual linear superposition rule of quantum mechanics with a
M\"obius-type ``composition law'' $\qop$, motivated by (i) bounded-domain composition laws in special
relativity, (ii) familiar transfer-matrix formulas in multilayer optics, and (iii) an analogy with
the inclusion--exclusion rule for classical probabilities.
In this note we explain why the proposal does not work as a modification of quantum theory.
For two components, the new rule differs from the ordinary sum only by an overall scalar factor,
so after normalization it represents the same ray and cannot change any physical prediction.
For three or more components, if one extends the two-term prescription in the natural recursive
way, the result becomes bracket/order dependent and can even change the ray, so a ``state'' is no
longer uniquely determined by a given preparation.
We also clarify why the inclusion--exclusion argument and the optics analogy do not support a
foundational change to Hilbert-space linearity.

\begin{keywords}
quantum superposition; Hilbert-space linearity; M\"obius transformation; deformed addition;
non-associativity; ray equivalence; Born rule; interference; Fabry--Perot resummation;
Lorentz group; Wigner (Thomas) rotation
\end{keywords}

\section{Introduction: ``deformations'' as representations vs.\ new physics}

A common pattern in ``generalized arithmetic'' programs is this.
One starts from a perfectly legitimate bounded-domain composition law---for example Einstein
velocity addition on $(-c,c)$, disk automorphisms, or the fractional-linear formulas that appear in
optics via transfer matrices---and then promotes that nonlinear operation to a \emph{replacement}
for the linear structure that the original theory relies on.
The new framework can look coherent on paper, but it often obtains its headline claims by quietly
changing premises or by changing what is meant by the experimentally aggregated quantities.

This methodological point is familiar from debates about Czachor-type deformed calculi in Bell/CHSH
discussions: if one replaces ordinary additivity of expectations by a deformed operation (linear
only with respect to a chosen $\oplus_f$), then the standard CHSH algebra no longer goes through.
But that is a \emph{change of framework}, not a refutation of Bell.
A useful diagnostic is an ``admissibility checklist'' for any proposed deformation:
(i) regularity (no singularities or ill-defined limits),
(ii) closure (the operation is defined on the stated domain),
(iii) symmetry (the deformation matches the genuine physical symmetry), and
(iv) an operational bridge (a clear device-level rule linking laboratory aggregation to the
deformed operation, rather than to ordinary $+$).

We apply the same lens here to Vigoureux's proposal in arXiv:2601.04248v1 to replace quantum
superposition by a M\"obius-type composition law $\qop$ \cite{Vigoureux2026}.

We use the standard physics convention for the Hilbert-space inner product:
$\braket{v}{w}$ is conjugate-linear in the first slot and linear in the second.
Pure states are identified with rays (nonzero vectors modulo nonzero complex scalars).

\section{What Vigoureux proposes}

The paper argues that ``bounded quantities'' should not be combined by ordinary addition $+$ but by
a M\"obius-type composition law $\qop$ (unit-disk automorphisms), and that this law should be used
in quantum theory \cite{Vigoureux2026}.
Concretely, for two states it proposes replacing
\begin{equation}
\ket{\psi} = c_1\ket{\phi_1} + c_2\ket{\phi_2}
\end{equation}
by
\begin{equation}\label{eq:vig-oplus}
\ket{\psi}
= c_1\ket{\phi_1}\qop c_2\ket{\phi_2}
:= \frac{c_1\ket{\phi_1}+c_2\ket{\phi_2}}{1+c_1^*c_2\braket{\phi_1}{\phi_2}}.
\end{equation}
The author suggests that the denominator reflects ``non-exclusivity'' before measurement and that
using $+$ ``overestimates'' probabilities \cite{Vigoureux2026}.
The manuscript also remarks that, since the denominator is a scalar, the resulting object ``also is
a solution of the Schr\"odinger equation'' \cite{Vigoureux2026}.

\section{Fatal point \#1: for two components the proposal is ray-trivial (or contradicts its own claims)}

\subsection{Ray equivalence: $\qop$ differs from $+$ only by an overall scalar}

Let
\begin{equation}
\ket{\psi_{+}} := c_1\ket{\phi_1}+c_2\ket{\phi_2},
\qquad
s := 1+c_1^*c_2\braket{\phi_1}{\phi_2}.
\end{equation}
Then \eqref{eq:vig-oplus} is simply
\begin{equation}\label{eq:scalar-rescale}
\ket{\psi_{\qop}} = \frac{1}{s}\,\ket{\psi_{+}}.
\end{equation}
Whenever $s\neq 0$, $\ket{\psi_{\qop}}$ is proportional to $\ket{\psi_{+}}$.
But in quantum mechanics physical pure states are rays:
\(
\ket{\psi}\sim \lambda\ket{\psi}
\)
for any nonzero $\lambda\in\C$.
So, for two components, the $\qop$ prescription cannot change any prediction \emph{as long as one
normalizes in the usual way}.
In particular, overall rescalings cancel from normalized probabilities:
\begin{equation}\label{eq:born-invariant}
\frac{\lvert \bra{x}\psi_{\qop}\rangle\rvert^2}{\braket{\psi_{\qop}}{\psi_{\qop}}}
=
\frac{\lvert \bra{x}\psi_{+}\rangle\rvert^2}{\braket{\psi_{+}}{\psi_{+}}}.
\end{equation}
This is hard to reconcile with the manuscript's repeated motivation that the denominator shows
probabilities computed with $+$ are ``too large'' \cite{Vigoureux2026}.

\subsection{The denominator is not a normalization factor}

If the intent is to ``build in'' normalization, the denominator does not do that.
Normalization depends on $\sqrt{\braket{\psi_{+}}{\psi_{+}}}$, not on $s$.
Indeed,
\begin{equation}\label{eq:norm}
\|\psi_{\qop}\|^2
=
\frac{\|\psi_{+}\|^2}{|s|^2}
=
\frac{|c_1|^2+|c_2|^2+2\Re\!\big(c_1^*c_2\braket{\phi_1}{\phi_2}\big)}
{\big|1+c_1^*c_2\braket{\phi_1}{\phi_2}\big|^2},
\end{equation}
which is not identically $1$.

\subsection{Ill-definedness: poles and cancellation limits}

The binary rule is undefined when the denominator vanishes.
In the two-state form \eqref{eq:vig-oplus} this occurs if
\begin{equation}\label{eq:pole-two}
1+c_1^*c_2\braket{\phi_1}{\phi_2}=0,
\end{equation}
and in the vector form \eqref{eq:vector-oplus} if
\begin{equation}\label{eq:pole-vector}
1+\braket{v}{w}=0,
\end{equation}
even when $v+w\neq 0$.

A further issue appears in a basic cancellation regime.
If $\ket{\phi_1}=\ket{\phi_2}=: \ket{\phi}$ and $c_2=-c_1$, the usual sum gives the zero vector
$\ket{\psi_{+}}=\mathbf{0}$, while the $\qop$ expression gives
\begin{equation}\label{eq:00}
\ket{\psi_{\qop}}
=
\frac{(c_1-c_1)\ket{\phi}}{1-|c_1|^2},
\end{equation}
so for $|c_1|=1$ it becomes an indeterminate $0/0$.
The point is not that $\mathbf{0}$ should represent a physical state; it is that the proposed
formula introduces poles/indeterminacies precisely where linear superposition is continuous and
well behaved, so extra limiting prescriptions would be needed.

\section{Fatal point \#2: for three or more components the rule is bracket/order dependent}

The manuscript primarily states the two-term prescription \eqref{eq:vig-oplus}.
To discuss $N\ge 3$ components one must extend it.
A natural extension---and the one suggested by the paper's emphasis on ``weak associativity'' and
bracket ordering---is to define, for general vectors,
\begin{equation}\label{eq:vector-oplus}
v \qop w := \frac{v+w}{1+\braket{v}{w}},
\end{equation}
and then to build multi-term ``superpositions'' recursively with explicit parentheses.
If one instead stipulates a different multi-ary rule that restores associativity or removes bracket
dependence, then one is adding \emph{extra structure not specified in the manuscript}, i.e.\ one is
defining a different theory.

With the natural recursive extension, the problem becomes immediate:
\emph{a physical preparation does not come with an intrinsic parenthesization convention}.
Moreover, the ambiguity is not merely a global phase.

\begin{prop}[Bracket dependence can change the ray]\label{prop:nonassoc}
There exist $u,v,w$ in a finite-dimensional Hilbert space such that
$(v\qop w)\qop u$ and $v\qop(w\qop u)$ are not proportional, hence represent different rays.
\end{prop}

\begin{proof}
Work in $\C^2$ with the computational basis
\(
\ket{0}=(1,0)^{\mathsf T},
\ket{1}=(0,1)^{\mathsf T}
\)
and
\(
\ket{+}=(\ket{0}+\ket{1})/\sqrt2.
\)
Let
\begin{equation}\label{eq:uvw}
v=0.8\ket{0},\qquad
w=0.8\ket{+},\qquad
u=0.8\ket{1},
\end{equation}
and define $\qop$ by \eqref{eq:vector-oplus}.
A direct calculation gives
\begin{equation}\label{eq:vw}
v\qop w = (0.94019964,\,0.38944344)^{\mathsf T},
\end{equation}
\begin{equation}\label{eq:leftbracket}
(v\qop w)\qop u = (0.71685886,\,0.90689576)^{\mathsf T},
\end{equation}
\begin{equation}\label{eq:rightbracket}
v\qop (w\qop u) = (0.90689576,\,0.71685886)^{\mathsf T}.
\end{equation}
These two vectors are not proportional (their component ratios differ), so they represent different
rays.
\end{proof}

For the binary operation \eqref{eq:vector-oplus}, for all pairs $(v,w)$ in its domain
(i.e.\ assuming $1+\langle v\mid w\rangle \neq 0$ and $1+\langle w\mid v\rangle\neq 0$),
one has
\begin{equation}\label{eq:comm-up-to-phase}
w\qop v
=
\frac{1+\braket{v}{w}}{1+\braket{v}{w}^*}\,(v\qop w),
\end{equation}
and the prefactor has unit modulus. Hence $v\qop w$ and $w\qop v$ differ only by a phase and
represent the same ray at the two-term level.
The genuine operational ambiguity arises for $N\ge 3$, where different bracketings can yield
non-proportional vectors (Proposition~\ref{prop:nonassoc}), hence different rays.

\section{The ``bounded amplitudes'' motivation is a category error}

The manuscript motivates $\qop$ by saying that ordinary $+$ can lead outside the unit disk, so $+$
is unsuitable for ``probability amplitudes'' \cite{Vigoureux2026}.
A careful formulation matters here:

\begin{itemize}
\item If a normalized state is expanded in an orthonormal basis,
$\ket{\psi}=\sum_n c_n\ket{n}$, then $\sum_n |c_n|^2=1$ and each $|c_n|\le 1$.
But this is a statement about a particular representation, not a basis-invariant principle that
would single out a disk-preserving ``addition'' law.

\item More invariantly, for any normalized $\ket{\psi}$ and any normalized $\ket{\phi}$,
Cauchy--Schwarz gives $|\braket{\phi}{\psi}|\le 1$.
This is a bound on transition amplitudes between normalized states; it does not imply that
arbitrary intermediate coefficients in nonorthogonal decompositions, path expansions, or unnormalized
partial contributions must remain inside the unit disk under a universal binary combination.

\item Quantum mechanics constrains the ray (state up to scalar) and, after normalization, the Born
probabilities. It does not require that coefficients in arbitrary decompositions remain in the unit
disk under some universal ``addition'' law.
\end{itemize}

Accordingly, boundedness of certain coefficient sets does not motivate replacing Hilbert-space
linearity; at most it motivates bounded parameterizations in specific subproblems.

\section{The inclusion--exclusion probability analogy does not carry to quantum events}

Vigoureux interprets the denominator via the classical identity
\begin{equation}\label{eq:incl-excl}
P(A\cup B)=P(A)+P(B)-P(A\cap B),
\end{equation}
rewriting it into a form reminiscent of $\qop$ and suggesting that the denominator encodes a
``joint probability'' term \cite{Vigoureux2026}.

Two problems are immediate:

\begin{itemize}
\item Quantum ``events'' are typically projectors, and noncommuting projectors do not form a single
Boolean algebra. In that setting, naive inclusion--exclusion identities are not derivations of
quantum amplitude-combination rules.

\item In \eqref{eq:vig-oplus} the denominator involves the complex quantity
$c_1^*c_2\braket{\phi_1}{\phi_2}$, which is not a real number in $[0,1]$ and so cannot be identified
with a classical joint probability.
\end{itemize}

\section{Why the optics analogy cuts the other way}

The manuscript emphasizes plane-parallel plates/Fabry--Perot interference, where reflection
amplitudes can be written in M\"obius form, and treats this as evidence for using $\qop$ in quantum
superposition \cite{Vigoureux2026}.
But the standard Fabry--Perot result is obtained by summing an infinite geometric series of
multiple reflections (repeated linear additions of path amplitudes), which then simplifies to a
fractional expression.
The moral is therefore the opposite:
\begin{quote}
M\"obius-type formulas can emerge \emph{within} linear wave theory as closed-form resummations of
many-path contributions; they do not indicate that superposition itself should be replaced.
\end{quote}

\section{The Schr\"odinger-equation remark is true only in a way that makes $\qop$ vacuous}

The manuscript states that because the denominator in \eqref{eq:vig-oplus} is a scalar, the $\qop$
superposition ``also is a solution of the Schr\"odinger equation'' \cite{Vigoureux2026}.
For the time-independent Schr\"odinger equation this is trivial: multiplying an eigenfunction by a
constant scalar does not change the eigenvalue equation.
For the time-dependent Schr\"odinger equation, the statement is correct only in a limited sense:
if both $\ket{\phi_1(t)}$ and $\ket{\phi_2(t)}$ evolve under the same unitary $U(t)$, then
$\braket{\phi_1(t)}{\phi_2(t)}$ is constant in time, so $s$ is constant and $\ket{\psi_{\qop}(t)}$
differs from $\ket{\psi_{+}(t)}$ only by a constant scalar and hence represents the same ray.

If one instead treats the scalar factor as time-dependent in a way not induced by unitary evolution
of the components, then multiplying by that factor will not, in general, preserve the
time-dependent Schr\"odinger equation. The manuscript does not provide a consistent dynamical
framework in which $\qop$ is both nontrivial and compatible with standard quantum dynamics.

The paper claims that using $\qop$ yields faster convergence, reduced ``noise,'' and improved
numerical stability through progressive ``renormalization,'' referring to earlier work
\cite{Vigoureux2026}. In the present manuscript, these claims are not supported by benchmarks,
controlled numerical comparisons, or an error analysis, and remain speculative as stated.

\section{Conclusion}

As a proposal to modify the foundational superposition rule of quantum mechanics, the manuscript
fails on two basic requirements:

\begin{itemize}
\item For two components, the $\qop$ rule produces a vector proportional to the usual sum and so is
physically ray-equivalent after normalization; yet the manuscript motivates it as a correction to
probabilities obtained with $+$.

\item For three or more components, once one adopts the natural recursive extension induced by the
two-term rule, the result is bracket/order dependent and can change the ray, so the ``state'' is
not uniquely determined by the preparation.
\end{itemize}

The probability and optics analogies do not repair these issues; in particular, M\"obius-type
formulas in optics arise from resummations within linear theory rather than from a replacement of
linear superposition.

\section*{Note added (related work on M\"obius-type ``additions'' and the 3D case)}

After completing the present manuscript, we became aware of earlier contributions by
J.-M.~Vigoureux that are directly relevant to the broader ``deformed addition'' narrative
surrounding Czachor-style proposals.

First, in the special-relativistic setting, Vigoureux reformulates the composition of non-parallel
velocities in the unit disk via a M\"obius (disk-automorphism) law (for dimensionless complex
velocities, i.e.\ in units $c=1$),
\(W=(V_1+V_2)/(1+\overline{V}_1V_2)\),
with the polar parametrization \(V=\tanh(a/2)e^{i\alpha}\).
He emphasizes that for multiple compositions one must iterate with an explicit bracketing order,
reflecting the nontrivial group structure behind non-collinear boosts.
He also stresses that the complex-plane formula is essentially tied to the coplanar case.
For genuinely three-dimensional, non-coplanar compositions he gives a Pauli-algebra / \(2\times 2\)
matrix formulation, which is intrinsically non-commutative and does not reduce to an ordinary
commutative/associative ``vector addition'' law.
More precisely, composing non-collinear boosts yields, in general, a boost together with a Wigner
(Thomas) rotation, and this is the obstruction to any reduction to a simple
commutative/associative ``addition'' on $\mathbb{R}^3$.%
\cite{VigoureuxEJP2013}

Second, in a recent preprint, Vigoureux advocates promoting essentially the same disk
automorphism/M\"obius-type composition law \(\oplus\) from relativistic velocity composition and
multilayer optics to a replacement for the linear superposition rule in quantum theory, again
explicitly warning about ``weak associativity'' and the need to track bracketing order when
composing many terms.
While we disagree with that quantum-theoretic extrapolation, the preprint provides independent
context showing that the M\"obius-composition viewpoint is not unique to Czachor and that the
3D/non-coplanar Lorentz case forces one beyond any ordinary commutative/associative ``vector
addition'' law into intrinsically non-abelian structure.
See Ref.~\cite{Vigoureux2026} and the references therein.

We were not aware of these Vigoureux references when preparing our discussion of Czachor's
velocity-addition arithmetic and its limitations; we add them here to clarify priority and to
situate our critique within the broader literature on M\"obius-type composition laws.


\begin{thebibliography}{9}

\bibitem{Vigoureux2026}
J.-M.\ Vigoureux,
\emph{Superposition of states in quantum theory},
arXiv:2601.04248v1 (2026).

\bibitem{VigoureuxEJP2013}
J.-M.\ Vigoureux,
\emph{A generalization in the complex plane of the composition law of non-parallel
velocities in special relativity},
Eur.\ J.\ Phys.\ \textbf{34} (2013) 795--803,
doi:10.1088/0143-0807/34/3/795.

\bibitem{CzachorNote}
Sienicki, M., and K. Sienicki.
\emph{Representations, Not Revolutions: Czachor's Calculus and Bell's Theorem.}
Acta Physica Polonica A \textbf{148}, no. 4 (2025).
\url{https://appol.ifpan.edu.pl/index.php/appa/article/view/148_273}

\bibitem{FeynmanHibbs}
R.\ P.\ Feynman and A.\ R.\ Hibbs,
\emph{Quantum Mechanics and Path Integrals},
McGraw--Hill (1965); reprints exist.

\bibitem{Sakurai}
J.\ J.\ Sakurai and J.\ Napolitano,
\emph{Modern Quantum Mechanics},
2nd ed., Addison--Wesley (2011).

\bibitem{NielsenChuang}
M.\ A.\ Nielsen and I.\ L.\ Chuang,
\emph{Quantum Computation and Quantum Information},
Cambridge University Press (2000).

\end{thebibliography}
\end{document}